\documentclass[a4paper,USenglish,thm-restate, cleveref]{lipics-v2021}
\nolinenumbers

\author{Fabien Dufoulon}
{Lancaster University, UK}
{f.dufoulon@lancaster.ac.uk}
{0000-0003-2977-4109}
{}
\author{Gopal Pandurangan}
{University of Houston, Texas, USA}
{gopalpandurangan@gmail.com}
{0000-0001-5833-6592}
{}
\author{Peter Robinson}
{Augusta University, Augusta, Georgia, USA}
{perobinson@augusta.edu}
{0000-0002-7442-7002}
{}
\authorrunning{F.\ Dufoulon, G.\ Pandurangan, and P.\ Robinson}
\Copyright{Fabien Dufoulon, Gopal Pandurangan, and Pandurangan Robinson}

\funding{Gopal Pandurangan was supported in part by Army Research Office (ARO) grant W911NF-231-0191 and National Science Foundation (NSF) grant CCF-2402837. Peter Robinson was supported in part by National Science Foundation (NSF) Grant No. CCF-2402836.}


\usepackage{graphicx} 
\usepackage{amsmath,amsbsy}
\usepackage{thmtools}
\usepackage{color}
\usepackage{comment}
\usepackage{xcolor}
\usepackage{algorithm}
\usepackage[noend]{algpseudocode}
\usepackage{multirow}
\usepackage{xspace}
\usepackage{threeparttable}
\usepackage{booktabs}
\usepackage{url}
\usepackage{subcaption}
\usepackage{mdframed}
\let\originalleft\left
\let\originalright\right
\renewcommand{\left}{\mathopen{}\mathclose\bgroup\originalleft}
\renewcommand{\right}{\aftergroup\egroup\originalright}

\renewcommand{\Pr}{\operatorname*{\textbf{\textup{Pr}}}}
\DeclareMathOperator*{\II}{\textbf{\textup{I}}}
\DeclareMathOperator*{\HH}{\textbf{\textup{H}}}
\DeclareMathOperator*{\EE}{\textbf{\textup{E}}}

\newcommand{\set}[1]{\{ #1 \}}

\newcommand{\lt}{\left}
\newcommand{\rt}{\right}
\newcommand{\md}{\middle}

\DeclareMathOperator{\poly}{\ensuremath{\mathrm{poly}}}

\newcommand{\congest}{\ensuremath{\mathsf{CONGEST}}}
\newcommand{\IC}{\ensuremath{\mathsf{IC}}}
\newcommand{\supp}{\ensuremath{\mathsf{supp}}}
\newcommand{\disj}{\ensuremath{\mathsf{DISJ}}}

\theoremstyle{plain}
\newtheorem{fact}{Fact}[section]

\EventEditors{Ioannis Chatzigiannakis, Andrea Vitaletti, Keren Censor-Hillel, and William K. Moses Jr.}
\EventNoEds{4}
\EventLongTitle{40th International Symposium on Distributed Computing (DISC 2026)}
\EventShortTitle{DISC 2026}
\EventAcronym{DISC}
\EventYear{2026}
\EventDate{November 9--13, 2026}
\EventLocation{Rome, Italy}
\EventLogo{}
\SeriesVolume{397}
\ArticleNo{13}

\newcommand{\ann}[1]{\text{\footnotesize(#1)}\quad}

\title{Tight Energy Lower Bounds  for Distributed Graph Algorithms}

\date{}

\begin{document}

\maketitle
\begin{abstract}
There has been a significant recent interest in designing distributed algorithms in the SLEEPING model that minimize the {\em energy (a.k.a awake)} complexity, which measures the number of rounds a node is {\em awake} during the  algorithm.  A node spends non-trivial resources
(messages, energy, etc.) only when it is awake and not while {\em sleeping}.
Energy complexity has been studied for various fundamental problems   with respect to minimizing the {\em maximum (worst-case)}  or the {\em average} number of rounds a node is awake. 

It has been shown that the energy complexities of several fundamental problems such as leader election (LE), broadcast, Minimum Spanning Tree (MST), Maximal Independent Set (MIS)  is {\em exponentially}  smaller compared to  their respective best-possible round complexities in the standard CONGEST model (where nodes can only send messages of small size). This raises a fundamental question of whether such significant energy gains are  possible for many other fundamental problems. 

Our main contribution is a general and powerful technique for showing energy lower bounds using information theory. It  gives almost a ``plug-in" way to show energy lower bounds for various problems in the standard CONGEST model.
Our information-theoretic technique allows us to leverage 
known lower bounds  on  {\em communication complexity} 
to obtain new, {\em almost optimal} (up to logarithmic factors) {\em polynomial (in $n$)}  lower bounds on energy complexity  --- for {\em both} worst-case and average-case ---
for fundamental graph problems such as triangle enumeration, All-Pairs Shortest Paths (APSP), diameter computation, minimum weight cycle,  Maximum Independent Set (MaxIS),
Minimum Dominating Set (MinDS),  Minimum Vertex Cover (MinVC). The energy lower bounds of these problems {\em match}
their respective  round lower bounds, implying that one {\em cannot} obtain any significant gains in energy complexity.
 \end{abstract}

\thispagestyle{empty}
\clearpage
\setcounter{page}{1}

\section{Introduction}
\label{sec:intro}

Motivated by  resource-constrained networks such as  sensor networks, where nodes  spend a lot of energy or other resources,  in recent years, there has been  significant research in designing  {\em energy-efficient} distributed algorithms for various fundamental problems (see e.g., \cite{energy1,CDHHLP18, CDHP20, podc2020,BM21, ghaffari-sleeping,ghaffari-podc2023, King_2011,DMP23,mst-wakeup,trygrub,dani-wakeup,Pierre,energy5,sleeping-coloring}). 
Such algorithms assume the {\em sleeping model} \cite{podc2020}, where nodes can
operate in two modes: \emph{awake} and \emph{sleeping}.  Each node can choose to
enter the awake or sleep state at the start of any specified round.
In sleeping mode, a node cannot send or receive, nor perform any non-trivial local computation, and it consumes very little energy or resources.
On the other hand, significant resources are utilized {\em only} in the awake mode, and hence the goal is to design distributed algorithms that {\em minimize} the number of awake rounds, i.e., the {\em energy (a.k.a awake)} complexity. It is challenging to design such algorithms, partly because one must judiciously balance keeping nodes asleep as much as possible with coordinating communication between nodes (we note that messages sent to a sleeping node are lost).

Energy complexity has been studied for various fundamental problems with respect to minimizing the {\em maximum (worst-case)} or the {\em average} number of rounds a node is awake.
While energy complexity is typically used to refer to worst-case scenarios, for average-case scenarios we use {\em node-averaged} energy complexity.
Both measures have been studied extensively for various problems (see e.g., \cite{podc2020, ghaffari-sleeping, DMP23} and the references therein).

Note that, traditionally, in  distributed algorithms, all nodes are considered awake in {\em all} rounds (i.e., there is no sleeping mode) and the goal is to minimize the \emph{round complexity}  of the algorithm, which counts
the (worst-case) {\em total} number of rounds taken by any node. 
After decades of intensive research, distributed algorithms with optimal (or near-optimal) round complexity for various fundamental problems are now well-established. In the standard CONGEST model of distributed computing, which captures bandwidth limitations of real-world networks, only messages of size $O(\log n)$ can be sent
over an edge per round.

Unfortunately, the round complexity of distributed algorithms can be quite large for various fundamental problems due to well-established lower bounds. Hence, several works have addressed designing algorithms that have {\em significantly less} energy (awake) complexity, even at the cost of increased round complexity, since only awake rounds contribute to energy.  To illustrate, for fundamental problems
such as leader election, broadcast, and spanning tree, it is well known that $\Omega(D)$ (where $D$ is the network diameter) is a universal lower bound on the round complexity \cite{jacm15},
but one can design algorithms with $O(\log n)$ energy complexity \cite{BM21,mst-wakeup}.
Similarly, for the fundamental Maximal Independent Set (MIS) problem, while $\Omega(\sqrt{\log n/\log \log n})$ is a lower bound on the round complexity \cite{Kuhn_2016}, there are algorithms with $O(\log \log n)$ energy complexity \cite{DMP23,ghaffari-podc2023}. Moreover, one can show that the {\em node-averaged} energy complexity of MIS is  $O(1)$, which is optimal \cite{podc2020,ghaffari-sleeping}.

Thus, the energy complexities of the above fundamental problems are {\em exponentially}  smaller compared to their respective best-possible round complexities. This raises a fundamental question of whether such large energy gains are possible for many other fundamental problems. 

The focus of this paper is to show {\em lower bounds} on the energy complexity of various fundamental graph problems. Generally, establishing lower bounds on energy complexity appears more difficult than establishing lower bounds on round complexity.
For one, a {\em  locality-based} approach (see e.g., \cite{jacm15,Kuhn_2016}) e.g., that captures many round complexity lower bounds does not apply (at least directly) for energy complexity.
Indeed, a basic property that underlies many round complexity lower bounds is that in $r$ rounds a node {\em cannot} get any information from {\em beyond} its $r$-hop neighborhood.  However, this property {\em does not} apply to energy complexity, since in $r$ {\em awake} rounds, a node can get information much farther than
its $r$-hop neighborhood. 

Communication complexity-based techniques give a uniform way to show
round complexity lower bounds in the CONGEST model (see e.g., \cite{stoc11, frischknecht2012networks,itcs24,keren-lb,keren-lb1, censor2017quadratic,bacrach2019hardness}).
At a high level, these techniques show that a lot of information --- usually established by communication complexity  ---    has to be exchanged across an appropriate {\em cut} in a graph, and if the cut size is small, this implies congestion across the cut edges, which leads to round lower bounds. This high-level idea does not directly apply to energy complexity, as it may be possible to transfer information simply by waking up at the appropriate rounds. Indeed, this is the reason why the communication complexity-based lower bound of $\tilde{\Omega}(D+\sqrt{n})$ for the fundamental Minimum Spanning Tree (MST) problem \cite{stoc11}
{\em does not} apply for energy complexity, and one can design an $O(\log n)$ energy algorithm \cite{mst-wakeup}.

\subsection{Our Contributions}
 We present a general and powerful technique for establishing energy lower bounds that applies not only to the worst-case setting but also to the node-averaged setting. As an application of our technique, we show new and (almost) tight energy complexity lower bounds for various fundamental problems. 
 
Using information-theoretic techniques, we present a technique to show energy lower bounds for various problems in the CONGEST model.
Our  technique allows us to leverage known lower bounds  on  {\em communication complexity}
to obtain new, {\em optimal}  {\em polynomial (in $n$)}  lower bounds on energy complexity  --- for {\em both} worst-case and average-case ---
for fundamental problems such as triangle enumeration, All-Pairs Shortest Paths (APSP), diameter computation,  Maximum Independent Set (MaxIS),
Minimum Dominating Set (MinDS), and Minimum Vertex Cover (MVC) (cf. Table \ref{fig:table}). The energy lower bounds of these problems {\em match}
their respective {\em round lower bounds}, implying that one {\em cannot} obtain any significant gains in energy complexity. More precisely, we show a lower bound of $\Omega(n^{1/3}/\log n)$ for triangle enumeration,  $\Omega(n/\log n)$ for both worst-case and node-averaged energy complexity for diameter computation, APSP, and minimum weight cycle finding.
We also show a worst-case energy lower bound of $\Omega(n^2/\log^2 n)$ and a node-averaged lower bound of $\Omega(n^2/\log^3 n)$  for MaxIS, MinDS, and MinVC. 

Our main technical result is summarized in the 
{\em Cut-based Energy Lower Bound Lemma}, which relates a {\em lower bound on the amount of information} that has to be transmitted across a cut  to a {\em lower bound on the energy complexity}:

\medskip
    \begin{mdframed}[innermargin =0.1cm,]
\noindent
\textbf{Lemma~\ref{lem:info_wc} (informal and simplified; see Sec.~\ref{sec:generic}, p.~\pageref{sec:generic}): Cut-based Energy Lower Bound.}
Suppose there exist two subsets of vertices $U$ and $\partial U$, such that each node in $U$ has at most $d_U$ neighbors in $\partial U$, and we are interested in some bit string $Z$ that is a function of the graph.
More concretely, the nodes in $U$ need to learn $L$ bits about $Z$, and this information is only known to the nodes in $\partial U$ initially.
Then, each node in $U$ must wake up at least $\Omega\lt( \frac{L}{d_U |U| \log n} \rt)$ times.
In fact, when $L$ roughly corresponds to the number of possibilities of $Z$, then we also obtain the same bound for the node-averaged energy complexity (over the nodes in $U$). 
These bounds hold independently of the number of rounds allowed before termination.
\end{mdframed}
\medskip

The lemma shows a precise relationship between the two lower bounds that depends on the amount of information, cut size, and the {\em degree} of nodes across the cut. The high-level idea
is to lower-bound the total number of messages that a node has to send across the cut and show that the node
has a small number of neighbors across the cut, which would imply a lower bound on the energy complexity. We use information theory to argue that the asymptotically same round lower bounds {\em also} apply to energy. 
However, there are two main technical challenges that we need to overcome:

First, we note that, unlike round complexity, for which it is easy to obtain a lower bound by dividing the information lower bound by the cut size, the situation is more subtle for energy complexity.
For instance, an algorithm could require an exponential number of rounds, $T\approx 2^{n^2}$, to compute the solution, whereby nodes wake up in only a small number of rounds to keep the energy complexity low. 
In particular, the algorithm could use the interval $[T]$ to encode the possible choices of the information that needs to be transmitted.\footnote{This approach is similar to time encoding to transmit information \cite{robinson2021being}, which can be used to reduce the message complexity at the cost of increasing the round complexity. In fact, one has to be careful in using communication complexity-based techniques for showing lower bounds in distributed computing with {\em synchronous clocks} \cite{cc-clocks, ssttheorem}.} 
If nodes on both sides of the cut could manage to wake up in the same particular round $i \in [T]$, then it may look plausible that the nodes on the receiving side of the cut have learned $\log T \approx n^2$ bits of information by being awake in just a single round.
The reader may rightfully object that achieving this coordination between the two sides of the cut may itself incur a high energy cost. 
To show formally that this (and related) strategies are futile, we make use of a result of Massey~\cite{massey1994guessing} that, intuitively speaking, allows us to lower bound the energy cost of ``guessing'' the right time to wake up for nodes on both sides.

The second challenge is combinatorial, where one has to appropriately set up a lower bound graph to ensure that nodes
have a small degree across the cut. We accomplish this by using $\ell$-separated graph families, which might be of independent interest (cf. Section \ref{sec:lsep}).

Our approach is general enough to apply to a broad range of graph problems, even though their existing lower-bound proofs for round complexity may rely on very different techniques. For instance, the hardness results of minimum vertex cover, maximum independent set, and minimum dominating set are all based on reductions from the set disjointness problem in 2-party communication complexity (see \cite{bacrach2019hardness}). 
In contrast, problems such as triangle listing or $s$-clique enumeration seem to require information-theoretic arguments from first principles~\cite{izumi2017triangle}. 
We show that our Cut-based Energy Lower Bound Lemma applies to any problem for which there is an $\ell$-separated graph family on which any algorithm exhibits high \emph{information cost}. This information cost is a notion from 2-party communication complexity introduced in \cite{DBLP:conf/stoc/BarakBCR10}, which, intuitively, quantifies how much information a protocol must leak about one party's input to the other.
This enables us to address all problems whose hardness is based on set disjointness, which is known to have high information cost. As a result, we directly obtain energy complexity lower bounds for numerous fundamental graph problems:
\medskip
\begin{mdframed}
\noindent\textbf{Theorem~\ref{thm:ic} (informal and simplified; see Sec.~\ref{sec:lsep} on page~\pageref{sec:lsep}.)}
Suppose that there exists an $\ell$-separated graph family with a cut set size $s$ and a maximum cut degree $d_{cut}$, such that solving a graph predicate $P$ exhibits information cost $I$.
Then, the worst case energy complexity is $\Omega\lt(\frac{I}{d_{cut}\cdot s \cdot \log n}\rt)$.
Moreover, assuming a natural bound on the inputs (satisfied by all problems whose round complexity hardness is based on set disjointness), the expected node-averaged energy complexity is 
$\Omega\lt(\frac{\ell}{n}\cdot \frac{I}{d_{cut} \log n}\rt)$.
\end{mdframed}
\medskip

Finally, we point out that $\ell$-separated graph families with large $\ell$ (say, $\ell = \Theta(n/\log n)$) are crucial for some of our results, specifically, for our quadratic node-averaged energy complexity lower bounds. Indeed, these quadratic lower bound instances exhibit small cut set size $s$ (say, $s = O(\log n)$). However, only cut nodes need to spend significant energy. To amplify the fraction of cut nodes (to $\ell s/n$), we ``stretch'' these small cuts by a large $\ell$ factor. This approach is inspired by that of \cite{itcs24}, in which stretching such cuts allows to obtain (time-conditional) cubic message complexity lower bounds.

\begin{table*}[t] \label{table: results}
\centering
\begin{threeparttable}
\begin{tabular*}{1.01\textwidth}{l l l l l}
\toprule
\textbf{Problem} & Lower Bound & Upper Bound & Our Result \\
\midrule
Triangle Listing & $\Omega\lt( \frac{n^{1/3}}{\log n} \rt)^{*}$ &  $\tilde{O}\lt( n^{1/3} \rt)^*$ \cite{CPSZ21} & Corollary~\ref{cor:triangle} \\
$s$-Clique Listing & $\Omega\lt( \frac{n^{1-2/s}}{\log n} \rt)^{*}$ &  $\tilde{O}\lt( n^{1-2/s} \rt)^*$ \cite{CCGL21} & Corollary~\ref{cor:scliques} \\
Local Triangle Listing & $\Omega\lt( \frac{n}{\log n} \rt)^{\dagger}$ &  $O(\frac{n}{\log n})^{\ddagger}$ \cite{HPZZ21} & Corollary~\ref{cor:triangle} \\
Computing Diameter / APSP & $\Omega\lt( \frac{n}{\log n} \rt)^{\dagger}$ & $O\lt( \frac{n}{\log n} \rt)^{*}$\cite{hua2016brief} & Corollary~\ref{cor:diameter}\\ 
Minimum Weight Cycle & $\Omega\lt( \frac{n}{\log n} \rt)^{\dagger}$ & $\tilde O\lt( n\rt)^{*}$\cite{manoharan} & Corollary~\ref{cor:diameter}\\
Minimum Vertex Cover & $\Omega\lt( \frac{n^2}{\log^2n} \rt)^{\dagger}$ & $O(n^2)^{*,\#}$ & Corollary~\ref{cor:itcs}\\
Maximum Independent Set & $\Omega\lt( \frac{n^2}{\log^2n} \rt)^{\dagger}$ & $O(n^2)^{*,\#}$ & Corollary~\ref{cor:itcs}\\
Minimum Dominating Set & $\Omega\lt( \frac{n^2}{\log^2n} \rt)^{\dagger}$ & $O(n^2)^{*,\#}$ & Corollary~\ref{cor:itcs}\\
\bottomrule
\end{tabular*}
{
  \begin{tablenotes}
  \item[$*$] Worst-case energy complexity.
  \item[$\dagger$] Expected node-averaged (which also implies worst-case) energy complexity.
  \item[$\#$] The trivial algorithm of learning the entire topology.
  \item[$\ddagger$] When parametrized also by the maximum degree $\Delta$, the upper bound of \cite{HPZZ21} is $O(\Delta / \log n + \log \log \Delta)$\\
  \end{tablenotes}
}
\caption{Energy Complexity Bounds for Graph Problems in the CONGEST model. All energy complexity lower bounds match the best-known round complexity lower bounds.} 
\label{fig:table}
\end{threeparttable}
\end{table*}

\section{Preliminaries}
\label{sec:model}
\subsection{Distributed Computing Model}
The (synchronous) CONGEST model~\cite{Peleg_2000_Book} considers some $n$-node input graph $G = (V,E)$, whose nodes represent the machines in a distributed network, and edges represent the communication links between any two machines. 
Each node has a unique integer ID of $O(\log n)$ bits, and we assume the \emph{clean network model}~\cite{Peleg_2000_Book}, where nodes are unaware of their neighbors' IDs initially. 
Note that, in the context of energy complexity, there is no substantial difference between the clean network model and the stronger $KT_1$ assumption~\cite{vainish}, where each node knows the IDs of its neighbors from the start:  
In the clean network model, each node can learn its neighbors' IDs by being awake in round $1$.  
In addition, we encode problem-specific inputs via some function $in$ defined on the vertices, where for any $v \in V$, $in(v)$ encodes the portion of the problem-specific inputs that pertain to $v$ and its incident edges. Finally, each node also has some common knowledge regarding $G$, such as a polynomial upper bound $N$ on the number of nodes, and has access to some private randomness (which is also independent of the nodes' inputs).

Computation proceeds in synchronous rounds. In each round, each node can send (possibly different) messages of $O(\log n)$ size to each of its neighbors, receives the messages sent by its neighbors, and performs some local computations. The {\em round complexity} of some algorithm is defined as the worst-case, over all nodes, number of rounds a node requires to produce an output and terminate.

\subsection{Sleeping Model} 
\label{sec:sleeping}
We assume the sleeping model~\cite{podc2020}, where 
 a node can be in
 either of the two states --- sleeping or awake.  
 (At the beginning, we assume that all nodes are awake.)  This is a simple generalization of the standard distributed computing model, where nodes are always assumed 
 to be awake.
 In the sleeping model, each node decides to be either \textit{awake} or \textit{asleep} in each round (till it terminates), corresponding to whether the node can receive/send messages and perform computations in that round or not, respectively. That is, any node $v$ can decide to {\em sleep} starting at any (specified) round of its choice. We assume that all nodes know the correct round number whenever they are awake. A node can {\em wake up} again later at any {\em specified} round and enter the {\em awake} state.
 We note that the model allows a node to cycle through the process of sleeping in some round and waking up in a later round as many times as it wishes. To summarize, distributed computation in the sleeping model proceeds 
 in \emph{synchronous} rounds, and each round consists of the following steps: (1) Each awake node can perform local computation. (2) Each awake node can send a message to its adjacent nodes.
    (3) Each awake node can receive messages sent to it in this round (in the previous step) by other awake nodes. 

   Thus, the paper assumes the sleeping model in the  CONGEST setting,  called the SLEEPING-CONGEST model~\cite{DMP23}. 
     
   \subparagraph{Energy Complexity.}
   In the sleeping model, let $A_v$ denote the number of awake rounds  for a node
    $v$ before it terminates (i.e., finishes the execution of the algorithm, locally). A node utilizes significant energy only when it is awake. 
    Hence we define the \emph{(worst-case) energy (a.k.a) awake complexity}
     as $A_{max} = \max_{v \in V}A_v$.
     Note that, for a randomized algorithm, $A_v$ will be a random variable.
       We also define $A_{avg} =\frac{1}{n}\sum_{v \in V} A_v$, i.e., the average of the $A_v$ random variables.  Then the expected {\em node-averaged energy complexity} of the randomized algorithm is
    $E[A_{avg}] =  E[\frac{1}{n}\sum_{v \in V} A_v] = \frac{1}{n}\sum_{v \in V} E[A_v]$, where the expectation is taken over the random coin choices of the algorithm.

In this paper, we study {\em lower bounds} on both the
worst-case and the node-averaged energy complexity. For worst-case energy complexity, we study worst-case (per node) energy lower bounds for (Monte-Carlo) algorithms that succeed with constant probability. For node-averaged energy complexity, we study lower bounds on the expected node-averaged energy complexity. 
We note that lower bounds for node-averaged energy complexity also imply the same for worst-case energy complexity.

\subsection{Basic Facts from Information Theory} \label{app:tools}

We recall some facts from information theory that we rely on in the proofs of Sections~\ref{sec:info} and \ref{sec:apps}.
Additional details and proofs can be found in standard textbooks such as \cite{clover_book}.

Here, we consider jointly distributed random variables $W$, $X$, $Y$, and $Z$,
and we follow the convention of using capitals for random variables and lowercase letters for values of random variables.
We denote the \emph{Shannon entropy of $X$} by $\HH\lt[ X \rt]$, which is defined as 
\begin{align}
\HH[ X ] = \sum_x \Pr[ X \!=\! x] \log_2(1 /\Pr[ X \!=\! x]). \label{eq:entropy}
\end{align}
The \emph{conditional entropy of $X$ conditioned on $Y$} is defined as
  \begin{align}\label{eq:cond_ent}
    \HH[ X \mid Y ] &= \EE_Y[ \HH[ X \mid Y \!=\! y] ]. 
  \end{align}
\noindent The \emph{conditional mutual information of $X$ and $Y$} is defined as 
  \begin{align}
    \II[ X : Y \mid Z ]
      = \EE_Z[ \II[ X : Y \mid Z \!=\! z ]] 
      &= \HH[ X \mid Z ] - \HH[ X \mid Y, Z ].
\label{eq:mutual_cond} 
  \end{align}

\begin{fact} \label{f:leq}
It holds that:
\begin{enumerate} 
\item 
$\II[ X : Y \mid Z ] \le \HH[ X \mid Z ] \le \HH[ X ]$,
\item 
$\II[ X : Y \mid Z ] \le \II[ X,W : Y \mid Z ]$.
\end{enumerate}
\end{fact}
\begin{fact}[Chain Rule of Mutual Information] \label{f:Ichain}
$\II\lt[ X : Y, Z \mid W \rt] = \II\lt[ X : Y \mid W \rt] + \II\lt[ X : Z \mid W,Y \rt]$.
\end{fact}

\begin{fact}[Chain Rule of Conditional Entropy] \label{f:Hchain}
$\HH\lt[ X,Y \mid W \rt] = \HH\lt[ X \mid W \rt] + \HH\lt[ Y \mid X,W \rt].$
\end{fact}
\begin{fact} \label{f:supp}
$\HH\lt[ X\rt] \le \log_2\lt(|\supp(X)|\rt),$ where $\supp(X)$ denotes the support of $X$. Equality holds if $X$ is uniformly distributed.
\end{fact}

\begin{fact} \label{f:len}
$\II\lt[ X : Y, Z \mid W \rt] \le |X|$.
\end{fact}

\begin{fact}[Data Processing Inequality, see Theorem 2.8.1 in \cite{clover_book}] \label{f:data_processing}
  If random variables $X$, $Y$, and $Z$ form the Markov chain $X \to Y \to Z$, i.e., the conditional distribution of $Z$ depends only on $Y$ and is conditionally independent of $X$, then $\II[ X : Y ] \ge \II[ X : Z ]$.
\end{fact}

\subsection{Additional Related Work}

We mainly focus on related works on lower bounds 
on energy complexity, though as mentioned in Section \ref{sec:intro}, 
several recent works  have designed  energy-efficient  distributed {\em algorithms} in the sleeping model for various fundamental problems including MIS, approximate  matching, maximal matching, coloring, and vertex cover, broadcast, spanning tree, breadth-first spanning tree, minimum spanning tree (MST), and single-source shortest paths \cite{energy1,CDHHLP18, CDHP20, podc2020,BM21, ghaffari-sleeping,ghaffari-podc2023, King_2011,DMP23,mst-wakeup,trygrub,dani-wakeup,Pierre,energy5,sleeping-coloring}).
While some of these works assume the SLEEPING-CONGEST model, some assume the SLEEPING-RADIO model (see e.g., \cite{DMP23} for works in this model), where
the radio network model is assumed. In this model, 
nodes can only broadcast to their neighbors, and there are interference constraints.
It is important to note that lower bounds for the SLEEPING-CONGEST also apply to SLEEPING-RADIO, and hence all the lower bounds in this paper also apply.

Few works have shown lower bounds on energy complexity. As mentioned earlier, there is a lack of general techniques for showing energy lower bounds, partly due to the reasons mentioned in Section \ref{sec:intro}. 
Still, a notable lower bound on energy complexity is the $\Omega(\log n)$ lower bound
for broadcast, which also applies for leader election and minimum spanning tree \cite{CDHHLP18,mst-wakeup}. This bound applies even to randomized algorithms and to the LOCAL model (where there are no restrictions on message sizes, unlike in CONGEST).
The work of \cite{BM21} showed a {\em deterministic} lower bound of $\Omega(\log n)$ \emph{node-averaged} energy complexity
for leader election on rings. This lower bound exploits the fact that $\Omega(n \log n)$ is a deterministic lower bound on the message complexity of leader election in rings, and as a result, on average, a node has to
transmit $\Omega(\log n)$ messages and hence be awake for so many rounds. This basic idea of exploiting per-node congestion is also used to demonstrate an energy-round trade-off in \cite{mst-wakeup}. 
We note that the approach of \cite{mst-wakeup} for showing energy-round trade-off  (which uses the lower approach of  \cite{stoc11} in conjunction with per-node congestion) is different from the 
information-theoretic approach (cf. Lemma~\ref{lem:info_wc}) used in the current paper. 

Additionally, some energy complexity lower bounds were shown in the SLEEPING-RADIO model (see e.g., \cite{CDHHLP18,mis-radio}), but these lower bounds leverage {\em the radio nature} (in particular, collision constraints of the radio model) and {\em do not}  apply to CONGEST. For example, the work of \cite{CDHHLP18} provides two  polynomial in $n$ lower bounds for computing the graph diameter in the SLEEPING-RADIO model ---
$\Omega(n)$ energy is required to 
$(2-\epsilon)$-approximate diameter on dense graphs, and $\Omega(n/\log^2 n)$
is required to 
$(3/2 -\epsilon)$-approximate diameter on sparse graphs. While the above results hold for arbitrary graphs, subsequently it was shown that the energy bound can be improved to $O(\sqrt{n})$ in bounded genus graphs for computing the exact diameter and global minimum cut size in {\em bounded genus} graphs \cite{chang-diameter}. This paper also shows $\Omega(n)$ energy lower bound on computing the minimum cut size in unit-disc graphs and $s-t$ minimum cut size in 
planar bipartite graphs. It should be noted, as mentioned earlier, that all the above-mentioned (polynomial in $n$) lower bound results are due to the radio nature  (in fact, they do not depend on the bandwidth constraint) and {\em do not} apply to SLEEPING-CONGEST.

We note that, for SLEEPING-CONGEST, the above (unconditional\footnote{That is, it does not depend on round complexity, unlike the trade-off result in \cite{mst-wakeup}.}) energy lower bounds are logarithmic in $n$. No significantly higher energy lower bounds—such as a polynomial in $n$—were known for well-studied graph problems prior to our work. We present the first such lower bounds for various fundamental problems.

\section{An Information-Theoretic Framework for Energy Complexity Lower Bounds}  \label{sec:info}
In this section, we present a framework for deriving lower bounds on the energy complexity of distributed graph problems via information-theoretic techniques.
In Section~\ref{sec:generic}, we give the fairly generic Lemma~\ref{lem:info_wc} that captures 
the impact of the amount of information that an algorithm needs to send across a cut on the (expected) total number of awake rounds of the nodes that lie at the boundary of the cut. 
Then, in Section~\ref{sec:lsep}, we turn our attention to so-called $\ell$-separated lower bound graph families, which were originally introduced for proving round and message complexity lower bounds via problems in communication complexity, such as set disjointness. 
By leveraging Lemma~\ref{lem:info_wc}, we give Theorem~\ref{thm:ic} and quantify the relationship between the information cost of a communication complexity problem --- the amount of information that any protocol must leak about the players' inputs --- and the energy complexity of a distributed algorithm that can be simulated in the two-party communication complexity model.

\subsection{A Cut-based Approach to Energy Complexity Lower Bounds} \label{sec:generic}

Throughout this section, we use basic facts from information theory that we summarize in Section~\ref{sec:model} for completeness. 
In particular, for random variables $X$, $Y$, and $Z$, we use $\II[ X : Y \mid Z]$ to denote the conditional mutual information, i.e., how much information $X$ reveals about $Y$ (and vice versa) given $Z$.
Before stating our main technical lemma, we first introduce some notation.
For a node $u$ and a set $W$, we use $e(u,W)$ to denote the number of neighbors of $u$ in $W$.
When considering a probability distribution $\mathcal{G}$ on a family of graphs, we slightly abuse notation and write $G \in \mathcal{G}$ if $G$ has nonzero probability of being sampled from $\mathcal{G}$.
We define the \emph{support of a random variable $X$} as $\supp(X) = \set{ x \mid \Pr\lt[ X=x \rt] > 0}$.

\begin{figure}
\centering
\begin{subfigure}{0.45\textwidth}
\includegraphics[scale=1.1]{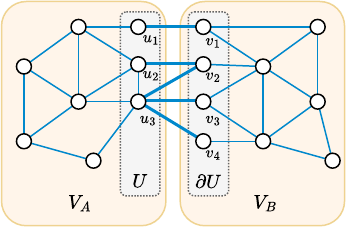}
\caption{\small
The cut-based lower bound approach. Lemma~\ref{lem:info_wc} gives a bound on the total number of awake rounds of the nodes in $U$. 
}
\label{fig:lb1} 
\end{subfigure}\quad
\hfill
\begin{subfigure}{0.5\textwidth}
\small
\centering
    \begin{tabular}{l l l}
        Round & Awake Nodes & Communication ($\partial U \to U$)  \\
        \hline
        2 & $u_1, u_2, v_1, v_2$ &  $v_1 \overset{m_1}{\to} u_1$, $v_2 \overset{m_2}{\to}
        \set{u_2,u_3}$\\
        5 & $u_1,u_3,v_1,v_3$  &  $v_1 \overset{m_3}{\to} u_1$,  $v_3 \overset{m_4}{\to} u_3$\\
            \hline\\ 
    \end{tabular}
\small
\vspace{-0.2cm}
    \begin{align*}
    \hat{\Pi}=
    \set{\
        &(2,v_1 \overset{m_1}{\to} u_1) \\
        & (2, v_2 \overset{m_2}{\to} u_2 ), \\
        &(5, v_1 \overset{m_3}{\to} u_1 ),\\
        & (5, v_3 \overset{m_4}{\to} u_3)\ 
    }
    \end{align*}
    \caption{The table shows a possible execution on the network in Figure~\ref{fig:lb1} and the resulting extended transcript $\hat{\Pi}$ (see Def.~\ref{def:extended_transcript}). For simplicity, we do not distinguish between a node and its ID in this example.}
    \label{fig:extended_transcript}
\end{subfigure}
\caption{}
\label{fig:lb}
\end{figure}
 
\begin{lemma}[Cut-based Energy Lower Bound] \label{lem:info_wc}
Consider a randomized $\epsilon$-error algorithm $\mathcal{A}$ in the sleeping $\congest$ model, and let $\mathcal{G}$ be a probability distribution of $n$-node graphs such that every $G \in \mathcal{G}$ has the same set of nodes.
Let $Z$ be a random variable that is a function of $G$.

Suppose that there exist sets $V_A$ and $V_B$ such that $(V_A,V_B)$ is a cut of every $G$. 
Let $U \subseteq V_A$ and $\partial U \subseteq V_B$ be any sets such that the edges in $(U,\partial U)$ form an edge cut, and every node in $U$ has a neighbor in $\partial U$. 
Let $d_U=\max_{G \in \mathcal{G}}\max_{u \in U}e(u,\partial U)$ be the maximum number of such neighbors.
Moreover, the choice of $V_A$ and $V_B$ must be independent of the algorithm's execution. 

Finally, let $\Pi_U$ be the transcript of the messages sent to the nodes in $U$ across $(U,\partial U)$.
If, for some $L=O(\poly(n))$, we have $\II\lt[\Pi_U :  Z  \ \md|\ S_{V_A} \rt] \ge L$,
where $S_{V_A}$ is the initial state of the nodes in $V_A$, 
then the following hold:
\begin{enumerate} 
\item[(i)] Let $\alpha_U$ be the number of total awake rounds summed over the nodes in $U$. \\ If $|\supp(Z)| \le 2^{c\cdot L}$, for some positive constant $c$, then
$\EE\lt[ \alpha_U \rt] = \Omega\lt(\frac{L}{d_U \log n}\rt)$.
\item[(ii)] The worst-case energy complexity bound of $\mathcal{A}$ is $ \Omega\lt(\frac{L}{d_U |U| \log n}\rt)$.
\end{enumerate} 
\end{lemma} 

Figure~\ref{fig:lb1} shows how the different node sets and the cut relate to each other. 
\subparagraph{Overview of the Proof of Lemma~\ref{lem:info_wc}.}
We first outline the high-level ideas and provide the full details in Section~\ref{sec:info_wc_proof}.
There are several technical challenges that we need to overcome. 
First, we do not impose any restriction on the algorithm's round complexity.
This, combined with the fact that nodes have access to a global clock and know the current round number, may lead us to believe that the algorithm could employ a time-encoding trick, where information of $L$ bits is conveyed across the cut by sending just a single bit in a specific round of an interval of length $2^L$, thus drastically reducing the number of required awake rounds.\footnote{In fact, it is known~\cite{robinson2021being} that time-encoding can reduce the \emph{message complexity} for many graph problems from $\Theta\lt( n^2 \rt)$ to near-optimal $O(n \poly\log n)$ messages at the cost of increasing the number of rounds.}
If successful, then the algorithm would be able to convey $L$ bits of information across an edge in just a single awake round of the sender and receiver.

We now provide some intuition on how we rule out such strategies in our proof. 
For simplicity, we limit our discussion to cuts involving a single sender and a single receiver.
Suppose that a node on the receiving side of the cut is indeed able to learn a large amount of information from the sender while waking up only a few times (and hence receiving only a few $O(\log n)$-size messages).
By carefully analyzing the \emph{extended transcript} of the communication across the cut, which augments each message with its round number and also includes the wake-up steps in which no message was received, we show that the receiver must have been able to correctly ``guess'' the round number of the sender's next wake-up. 
Furthermore, the sender must have used a sufficiently large range when selecting the time of their next awake round to ensure that the round number has the necessary amount of entropy.
As sender and receiver have to both be awake in the same round for a successful transmission to occur, we conclude that this is unlikely to happen unless the receiver wakes up a large number of times, by leveraging the relationship between the entropy of a random variable and the expected number of guesses for finding its value due to Massey~\cite{massey1994guessing}.
Intuitively, Massey's result rules out a time-encoding approach in which sending only a single message in a specific round, chosen from an exponentially large interval, can convey a large amount of information (and thus save on messages).
We obtain a contradiction on the assumed upper bound on the worst-case number of awake rounds, which shows Property~(ii) of the lemma.

To extend the proof to Property~(i), i.e., the case where the algorithm only satisfies the energy complexity bound in expectation, we first observe that a given node-averaged algorithm will satisfy (a constant factor of) the expected bound on the total awake complexity with constant probability, and apply the above argument directly to that case. 
However, when the algorithm exceeds the expected number of awake steps, it becomes more challenging to directly bound the amount of information that the receiver may learn. 
This is because conditioning on the event that the algorithm significantly exceeds the expected total number of awake rounds may actually increase how much useful information the algorithm's transcript carries about the random variable $Z$, which is a function of the subgraph on the sender's side that the receiver wants to learn.
We overcome this issue by introducing a restriction on the number of possible choices for $Z$ (i.e., its support). 
This essentially guarantees that the amount of information learned about $Z$ cannot be off by more than a constant factor compared to the case when the algorithm is close to the expected number of wake-up steps, which suffices to prove (i).

\subsection{Lower Bounds for \texorpdfstring{$\ell$}{l}-Separated Graph Families} \label{sec:lsep}

Many existing lower bounds for fundamental graph problems are based on the hardness of solving certain functions in the two-party model of communication complexity, where Alice and Bob have inputs $X$ and $Y$ respectively, and need to communicate in order to compute $f(X,Y)$ with error at most $\epsilon$.
To make it easier to apply Lemma~\ref{lem:info_wc} to such problems, we first introduce the notion of $\ell$-separated families of lower bound graphs that were defined in \cite{itcs24},
which are a generalization of the lower bound graph families of \cite{censor2017quadratic} and \cite{bacrach2019hardness}. 
Intuitively speaking, a graph $G$ of an $\ell$-separated family consists of $\ell$ vertex sets $V_1,\dots,V_\ell$ with the property that we can simulate a given distributed algorithm in the two-party model of communication complexity for any possible cut $(V_{\le i}, V_{>i})$ ($i \in [1,\ell-1]$). 
Moreover, $G[V_1]$ depends on Alice's input, and Bob uses his input to create $G[V_\ell]$, with the goal of simulating a distributed algorithm on $G$ to compute a given two-party function $f$.

We use the \emph{internal information cost} ($\IC$) introduced in \cite{DBLP:conf/stoc/BarakBCR10}  to quantify the amount of information that the cross-cut transcript of such a simulation must leak to Alice about Bob's input and vice versa. 
In Definition~\ref{def:gic} below, we extend the notion of information cost  to graph predicates $P$, such that $\IC(P)$ corresponds to  $\IC(f)$ of the function $f$ that exhibits the highest information cost of all functions that can be computed by simulating an algorithm for $P$ on some $\ell$-separated graph family.
We obtain Theorem~\ref{thm:ic}, which, intuitively speaking, states that a graph predicate with a high information cost entails a corresponding lower bound on its energy complexity.
To see why Theorem~\ref{thm:ic} holds, we observe that, when simulating an algorithm for $P$ for a given cut $(V_{\le i}, V_{>i})$, the nodes on one side of the cut must learn at least $\frac{\IC(P)}{2} \ge \frac{\IC(f)}{2}$ bits. 
Consequently, we can instantiate Lemma~\ref{lem:info_wc} for obtaining a lower bound on the expected total number of awake rounds of the corresponding nodes.
By summing over the $\ell-1$ possible cuts, the lower bound on the expected total energy complexity follows.
Figure~\ref{fig:lb2} shows the general setup used in Theorem~\ref{thm:ic} for the special case $\ell=3$.

\begin{definition}[$\ell$-Separated Family of Lower Bound Graphs; see Def~2.1 in \cite{itcs24}]\label{def:lb-graph-family}
    Let $f:X \times Y \to \{\mathsf{true},\mathsf{false}\}$ be a function and $P$ be a graph predicate. For an integer $\ell>1$, a family of graphs $\{G_{x,y}=(V,E_{x,y})\mid x \in X, y \in Y\}$ is said to be an \emph{$\ell$-separated family of lower bound graphs w.r.t.\ $f$ and $P$} if $V$ can be partitioned into $\ell$ disjoint subsets $V_1, V_2,\dots, V_{\ell-1}, V_\ell$ such that the following properties hold:
\begin{enumerate}
  \item[(a)] \label{lb-fw:alice} Only the existence or the weight of edges in $V_1 \times V_1$ depend on $x$;
  \item[(b)] \label{lb-fw:bob} Only the existence or the weight of edges in $V_\ell \times V_\ell$ depend on $y$;
  \item[(c)] \label{lb-fw:cuts} For all $1 \le i \le \ell$, the vertices in $V_i$ are only connected to vertices in $V_{i-1} \cup V_i \cup V_{i+1}$ (where $V_0 = V_{\ell+1} = \emptyset$). 
  \item[(d)] \label{lb-fw:pred} $G_{x,y}$ satisfies the predicate $P$ iff $f(x,y) = \mathsf{true}$.
\end{enumerate}
\end{definition}

The equivalence captured by (d) motivates simulating a distributed algorithm for the predicate $P$ in the two-party model of communication complexity, where Alice and Bob want to compute the value of $f(X,Y)$. 
As the simulation follows the standard approach used in previous work, we provide only a high-level overview; additional details can be found in \cite{itcs24,censor2017quadratic}.

\subparagraph{Simulation Protocol.}
For an index $i \in [\ell-1]$, we obtain the \emph{simulation protocol for cut $(V_{\le i},V_{\ge i+1})$}, by instructing Alice to construct the edges of $G[V_1]$, according to her input $X$, whereas Bob is responsible for constructing $G[V_\ell]$ according to $Y$.
The rest of $G$ does not depend on their inputs and hence can be (locally) constructed by both.
Then, they jointly simulate algorithm $\mathcal{A}$ round by round, whereby Alice simulates all nodes in $V_{\le i}$, and Bob simulates the remaining nodes given by $V_{\ge i+1}$.
Alice and Bob only need to communicate for the messages that cross the cut $(V_{\le i},V_{\ge i+1})$.
It is straightforward to verify that the simulation protocol has the same error probability as $\mathcal{A}$.

\subparagraph{Maximum Cut Degree and Cut Size.}  For each $i \in [\ell]$, we define subsets $V_{i}^-,V_{i}^+ \subseteq V_i$ as follows: $V_i^-$ contains all nodes in $V_i$ that have a neighbor in $V_{i-1}$ and $V_{i}^+$ consists of all nodes having a neighbor in $V_{i+1}$.  
  We call $s$ the \emph{cut set size} of the graph family and define $s = \max\limits_{1 \le i \le \ell}\lt( |V_i^-|,|V_i^+| \rt)$. 
We also define the \emph{maximum cut degree} 
\begin{align}
d_{cut}= \max\limits_{1 \le i \le \ell-1}\lt( \max\limits_{u \in V_i}e(u,V_{i+1}), \max\limits_{v \in V_{i+1}}e(v,V_{i})\rt). \label{eq:dcut}
\end{align}

\subparagraph{Information Cost.}
A common way to quantify how much information the transcript $\Pi$ of all messages leaks to the players in the two-party model about each other's input is the \emph{(internal) information cost~\cite{DBLP:conf/stoc/BarakBCR10} of the protocol $\pi$ on distribution $\mu$}, defined as 
\begin{align}
    \IC_{\mu}(\pi) = \II\lt[ X : \Pi \ \md|\ Y \rt] + \II\lt[ Y : \Pi \ \md|\ X \rt].\notag
\end{align}
Moreover, the \emph{information cost of $f$ with error $\epsilon$} is defined as 
\begin{align}
    \IC(f,\epsilon) = \inf_{\pi} \max_{\mu} \IC_\mu(\pi), \notag
\end{align}
where the infimum is taken over all $\epsilon$-error protocols $\pi$ that compute $f$, and the maximum identifies the worst case input distribution $\mu$ with respect to the best protocol $\pi$.

\begin{definition}[Information Cost of Graph Predicate] \label{def:gic}
Consider a graph predicate $P$ and a function $f : X \times Y \to \set{\mathsf{true},\mathsf{false}}$. 
We say that \emph{$P$ has an information cost of $\IC(f,\epsilon)$ with error $\epsilon$ for $f$}, formally $\IC(P,f,\epsilon) = \IC(f,\epsilon)$, if there exists an $\ell$-separated graph family with respect to $f$ and $P$.
The \emph{information cost of predicate $P$ with error $\epsilon$} is defined as
\begin{align}
    \IC(P,\epsilon) = \sup\nolimits_f \IC(P,f,\epsilon),
    \notag
\end{align}
where the supremum is taken over all functions $f$, for which there exists an $\ell$-separated family w.r.t.\ $P$.
\end{definition}

\begin{theorem}[Information Cost $\longrightarrow$ Energy Complexity] \label{thm:ic}
Consider a graph predicate $P$ and some $\epsilon>0$. 
Suppose that $\IC(P,\epsilon)\ge I$,
where $I = O(\poly(n))$, i.e., there exists a suitable function $f:X\times Y \to \set{\mathsf{true},\mathsf{false}}$ and an $\ell$-separated family such that $\IC(P,f,\epsilon) = \IC(f,\epsilon) \ge I$. Let $s$ be the cut set size (see Def.~\ref{def:lb-graph-family}), and let $d_{cut}$ denote the maximum cut degree; see \eqref{eq:dcut}.
The following hold for any $\epsilon$-error algorithm for $P$:
\begin{enumerate} 
\item[(a)] If $|\supp(X)| \le 2^{c\cdot I}$ and $|\supp(Y)| \le 2^{c\cdot I}$, for some positive constant $c$, then the expected node-averaged energy complexity is $\Omega\lt(\frac{\ell-1}{n}\cdot \frac{I}{d_{cut}\cdot \log n}\rt)$.
\item[(b)] The worst case energy complexity (per node) is $\Omega\lt(\frac{I}{d_{cut}\cdot s\cdot \log n}\rt)$.
\end{enumerate}
Moreover, the energy complexity bounds hold independently of the round complexity of the algorithm.
\end{theorem}
\begin{figure}[t]
  \centering
\includegraphics[scale=1.0]{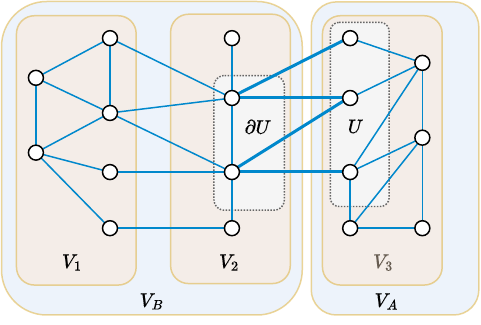}
\caption{\small
An example of the graph construction in the proof of Theorem~\ref{thm:ic} for a 3-separated graph family. 
For the node-averaged lower bound, we apply Lemma~\ref{lem:info_wc} to both cuts, $(V_1,V_{2} \cup V_3)$ and $(V_1 \cup V_2, V_3)$.
The figure depicts this application for the cut $(V_1 \cup V_2, V_3)$.
Here we assume that the nodes on $V_3$'s side learn more than $I/2$ bits about $G[V_1]$, and hence we choose the set $U$ to lie on this side and swap $V_A$ and $V_B$ accordingly when instantiating Lemma~\ref{lem:info_wc}.
}
\label{fig:lb2} 
\end{figure}

\begin{proof} 
We first show the bound on the worst-case energy complexity of a given algorithm $\mathcal{A}$.
Consider the $\ell$-separated graph family $\mathcal{G}$ w.r.t.\ $f$ and $P$, and fix some arbitrary $i\in [\ell-1]$.
Suppose that we use $\mathcal{A}$ to run the simulation protocol $\pi_i$ for the cut $(V_{\le i},V_{\ge i+1})$, for some $i \in [\ell-1]$, with the goal of computing $f(X,Y)$, and recall that Alice uses her input to construct the edges of $G[V_1]$, whereas Bob will construct $G[V_{\ell}]$ according to $Y$. 
Since $\IC(f,\epsilon) \ge I$, we know that there exists a distribution $\mu$ on $X\times Y$ such that the information cost of protocol $\pi_i$ is
\begin{align}
\IC_\mu(\pi_i) = \II\lt[ X : \Pi \mid Y \rt] + \II\lt[ Y : \Pi \mid X \rt]  \ge I.\notag
\end{align}
Without loss of generality, we assume that 
\begin{align}
\II\lt[ X : \Pi \mid Y \rt] \ge \frac{I}{2}. \label{eq:ic_lb1}
\end{align}
Let $\Pi_{V_{i}}$ denote the transcript of messages that the nodes in $V_{i}$ receive from their neighbors in $V_{i+1}$ and define $\Pi_{V_{i+1}}$ analogously; note that $\Pi=(\Pi_{V_{i}},\Pi_{V_{i+1}})$ contains the entire transcript of the communication between Alice and Bob, i.e., which corresponds to all messages crossing the cut.
Our goal is to apply Lemma~\ref{lem:info_wc}, and hence we need to first argue that  \eqref{eq:ic_lb1} holds when only considering the transcript $\Pi_{V_{i+1}}$.

\begin{claim} \label{cl:chainrule}
$\II\lt[ X : \Pi_{V_{i+1}} \mid  Y \rt] \ge \frac{I}{2}.$ 
\end{claim}
\begin{proof}
Let $R_B$ be the private randomness of Bob.
Consider the mutual information $\II\lt[ X : R_B \mid  \Pi,Y \rt]$.
Conditioned on the transcript $\Pi$ and Bob's input $Y$, we know that Alice's input $X$ is independent of Bob's randomness $R_B$. 
In other words, $X \to (\Pi,Y) \to R_B$ forms a Markov chain, and hence $\II\lt[ X : R_B \mid  \Pi,Y \rt]=0$.
By a similar argument, it follows that $\II\lt[ X : R_B \mid  Y \rt] = 0$.
Combining these observations and applying the chain rule multiple times yields
\begin{align}
 \II\lt[ X : \Pi \mid Y \rt]
 &=
 \II\lt[ X : \Pi \mid Y \rt]
 +
 \II\lt[ X : R_B \mid  \Pi,Y \rt]
 -
 \II\lt[ X : R_B \mid Y \rt] \label{eq:inf1}\\ 
 &=
 \II\lt[ X : \Pi, R_B \mid  Y \rt]
 -
 \II\lt[ X : R_B \mid Y \rt] \label{eq:inf2}\\ 
 &=
 \II\lt[ X : \Pi \mid  R_B,Y \rt] \label{eq:inf3}\\ 
 &=
 \II\lt[ X : \Pi_{V_{i}},\Pi_{V_{i+1}} \mid  R_B,Y \rt]\notag\\ 
 &=
 \II\lt[ X : \Pi_{V_{i+1}} \mid  R_B,Y \rt]
 +
 \II\lt[ X : \Pi_{V_{i}} \mid  \Pi_{V_{i+1}},R_B,Y \rt]\notag\\ 
 &=
 \II\lt[ X : \Pi_{V_{i+1}} \mid  R_B,Y \rt], \label{eq:ic_b1}
\end{align} 
where the final step follows because the conditioning on $\Pi_{V_{i+1}},R_B,Y$ ensures that the messages sent to the nodes in $V_i$, i.e., $\Pi_{V_{i}}$, do not reveal any information about $X$, and hence $\II\lt[ X : \Pi_{V_{i}} \mid  \Pi_{V_{i+1}},R_B,Y \rt] = 0$.
Finally, we can again remove the conditioning on the private randomness in \eqref{eq:ic_b1} by simply executing the chain rule applications in \eqref{eq:inf1}-\eqref{eq:inf3} in reverse order (after replacing $\Pi$ with $\Pi_{i+1}$).  
It follows that $\II\lt[ X : \Pi_{V_{i+1}} \mid  Y \rt] = \II\lt[ X : \Pi \mid Y \rt] \ge \frac{I}{2}$, which completes the proof of the claim.
\end{proof}
We are now ready to instantiate Lemma~\ref{lem:info_wc}(ii) for algorithm $\mathcal{A}$ with $Z=G[V_{1}]$, $L=\frac{I}{2}$, $V_A=V_{\ge i + 1}$, $V_B=V_{\le i}$, $U=V_{i+1}^-$, $\partial U = V_{i}^+$, and $d_U = d_{cut}$, which means that 
$\II\lt[ Z : \Pi_{U} \mid  S_{V_A} \rt] = \II\lt[ G[V_{1}] : \Pi_{V_{i+1}} \mid  S_{V_{\ge i+1}} \rt] \ge \frac{I}{2}$, as required by the premise of the lemma.
This yields the sought worst case bound per node.

To extend the result to the node-averaged case, we proceed similarly as for the worst case above, with the main difference being that we apply Lemma~\ref{lem:info_wc}(i) for every cut $(V_i,V_{i+1})$ ($1 \le i \le \ell-1$).
Let $\alpha_i$ be the total number of awake rounds of the nodes in $V_i \cup V_{i+1}$.
It follows that, for all $i \in [\ell-1]$, 
\[
\EE\lt[ \alpha_i \rt] = \Omega\lt( \frac{I}{d_{cut}\log n} \rt),
\]
and thus 
\begin{align}
\EE\lt[ \alpha  \rt] \ge \frac{1}{2}\sum_{i=1}^{\ell-1} \EE\lt[ \alpha_i  \rt] = \Omega\lt( \frac{(\ell-1)\cdot I}{d_{cut}\log n} \rt),\notag
\end{align}
and the claimed bound on the expected node-averaged complexity follows.
\end{proof}

The most prominent instance of a function $f$ in communication complexity is the $\emph{set disjointness}$ function ($\disj_m$): 
Alice and Bob each get a subset of some universe of size $m$, represented as the corresponding $m$-length characteristic vectors, and they must communicate to decide whether these sets are disjoint.
Numerous lower bound reductions employ the hardness result of set disjointness (e.g., \cite{censor2017quadratic, bacrach2019hardness}), which says that the players need to communicate $\Omega\lt( m \rt)$ bits.
For the subsequent applications, we require a stronger property, namely that any set disjointness protocol must leak a linear amount of information:

\begin{fact}[Information cost of set disjointness~\cite{dagan2017trading}] \label{f:disj}
The set disjointness function on input vectors of length $m$ has
$\IC(\disj_m,\epsilon) = \Omega\lt( m \rt)$, for a suitable small constant $\epsilon>0$.
\end{fact}

Next, we state a corollary that allows a streamlined application of Theorem~\ref{thm:ic} for several important graph problems.
That is, it enables us to directly obtain energy complexity bounds by reusing existing lower bound constructions in the $\congest$ model.

\begin{corollary} \label{cor:disj}
Consider a graph predicate $P$, and suppose that there exists an $\ell$-separated graph family $\mathcal{G}$ with respect to $\disj_m$ and $P$, where $m=\Theta\lt( n^2 \rt)$, and where $d_{cut}$ denotes the maximum cut degree and $s$ denotes the cut set size of $\mathcal{G}$.
Then, the worst case energy complexity of any distributed algorithm that decides $P$ is $\Omega\lt(\frac{n^2}{d_{cut}\cdot s \cdot \log n}\rt)$, whereas its expected node-averaged energy complexity is $\Omega\lt(\frac{\ell-1}{n}\cdot \frac{n^2}{d_{cut} \log n}\rt)$.
These bounds hold independently of the time complexity of the algorithm.
\end{corollary}
\begin{proof} 
The claim on the worst-case energy complexity is immediate from Fact~\ref{f:disj} and Theorem~\ref{thm:ic}(b).
For applying Theorem~\ref{thm:ic}(a), we need to argue that the support of Alice's input $X$ and Bob's input $Y$ are bounded from above by $2^{c \cdot n^2}$, which holds since $m=\Theta\lt( n^2 \rt)$.
\end{proof}

\section{Applications} 
\label{sec:apps}
We now demonstrate how to obtain lower bounds for a diverse range of graph problems via the framework from Section~\ref{sec:info}.
For showing a lower bound on the energy complexity of listing all triangles in the graph, we directly apply Lemma~\ref{lem:info_wc}, whereas, subsequently, we use Corollary~\ref{cor:disj} for graph problems, whose hardness is based on a reduction from set disjointness.

\begin{corollary}[Triangle Listing] \label{cor:triangle}
For \emph{triangle listing}, where every triangle in the graph is output by at least one of its nodes, the worst case energy complexity is $\Omega\lt( n^{1/3}/\log n \rt)$ rounds for algorithms that fail with some small constant probability $\epsilon>0$.
For \emph{local triangle listing}, which requires every node to output all triangles that it is part of, we obtain a lower bound of $\Omega\lt( n/\log n \rt)$ that holds even for the expected node-averaged energy complexity.
\end{corollary}
\begin{proof}
We make use of the existing lower bounds for the CONGEST model of \cite{izumi2017triangle,pandurangan2021distributed}. 
In more detail, we adapt the proof of Theorem~4.1 in \cite{izumi2017triangle}, where they show that the node who outputs the maximum number of triangles needs to learn $\Omega\lt( n^{4/3} \rt)$ bits (on average) about the edges $E$ of an Erd\"os-Renyi random graph $G_{n,1/2}$.

We use $N = \frac{n(n-1)(n-2)}{6}$ to denote the maximum possible number of triangles in any $n$-node graph. 
Define $w$ to be a node chosen uniformly at random (note that $w$ itself is a random variable), and let $T_w$ denote the set of triangles output by $w$.

In the proof of Theorem~4.1 in \cite{izumi2017triangle}, they show that, if $w$ is not a randomly chosen node, but instead the node that outputs the maximum number of triangles, then
\begin{align}
    \II\lt[ E : T_{w} \rt] \ge \frac{\sqrt{2}}{3} \lt( \frac{N}{16n} \rt)^{2/3} \cdot \Pr\lt[ |T_w| \ge \frac{N}{16n} \rt]. \label{eq:ig} 
\end{align}
We now argue that an asymptotically-equivalent bound also holds for a \emph{randomly} chosen $w$, under the assumption that each triangle needs to be output by at least one of its constituting nodes.
Since there are $\frac{N}{8}$ triangles in expectation, \cite{izumi2017triangle} show that there are at least $\frac{N}{16}$ triangles with probability at least $\frac{1}{15}$ by the (reverse) Markov's inequality. 
Condition on this event. 
Observing that each node is in at most $O(n^2)=O(N/n)$ triangles, implies that a constant fraction of the nodes must output at least $\Omega\lt( N/n \rt)$ each.
Since $w$ is chosen uniformly at random, it follows that $|T_w| \ge \Omega\lt( \frac{N}{n} \rt)$ happens with constant probability, and thus \eqref{eq:ig} tells us that
\begin{align} \label{eq:ig2}
		\II\lt[ E : T_{w} \rt] = \Omega\lt( \lt( \frac{N}{n} \rt)^{2/3} \rt).
\end{align}
Let $S_w$ denote the initial state of $w$, which includes the knowledge of its incident edges and let $\Pi_w$ be the transcript of the messages received by $w$.
We aim to apply Lemma~\ref{lem:info_wc} with $U=\set{w}$, and thus we need to derive a suitable lower bound on $\II\lt[ E : \Pi_w \ \md|\ S_w \rt]$.

Observe that $T_w$ is a function of $\Pi_w$, $S_w$, and the private randomness of $w$. 
By the data-processing inequality (see Fact~\ref{f:data_processing}) we obtain that
\begin{align}
	\II\lt[ E : \Pi_w \ \md|\ S_w \rt]
  &\ge
	\II\lt[ E : T_w \ \md|\ S_w \rt]\notag\\ 
  \ann{by \eqref{eq:mutual_cond}}
  &= \HH\lt[ E \ \md|\ S_w \rt]  - \HH\lt[ E \ \md|\ S_w, T_w \rt] \notag\\ 
  \ann{since $w$ has $\le n-1$ incident edges}
  &\ge \HH\lt[  E  \rt] - (n-1)- \HH\lt[ E \ \md|\ S_w, T_w \rt] \notag\\ 
  \ann{by Fact~\ref{f:leq}}
  &\ge \HH\lt[  E  \rt] - (n-1)- \HH\lt[ E \ \md|\ T_w \rt] \notag\\ 
  \ann{by \eqref{eq:mutual_cond}}
  &= \II\lt[  E : T_w  \rt] - (n-1)\notag\\ 
  \ann{by \eqref{eq:ig2}}
  &= \Omega\lt( n^{4/3} \rt) \label{eq:ig3}
\end{align}
To complete the proof of the worst-case energy complexity bound, we instantiate Lemma~\ref{lem:info_wc}(ii) with $V_A=U=\set{w}$, $V_B=\partial U = V(G)\setminus \set{w}$, $d_U=n-1$, $Z=E$, and $L= \Theta\lt( n^{4/3} \rt)$.
Note that there is no dependency between the algorithm's output and the choice of the cut $(V_A,V_B)$, as the latter only depends on the randomly chosen node $w$.
This shows that node $w$ has a worst-case energy complexity of $\Omega\lt( n^{1/3}/\log n \rt)$.

Next, we prove the node-averaged bound of $\Omega\lt( n / \log n \rt)$ for local triangle listing.
Let $\mathcal{P}(T_v)$ denote the set of edges that are part of the triangles output by node $v$.
From Lemma~4.3 of \cite{izumi2017triangle}, it follows that
$\II\lt[ E : T_{v} \rt] = \EE\lt[ |\mathcal{P}(T_v)| \rt]$, for any $v$.
Since $v$ needs to output all its triangles, \eqref{eq:ig2} can be strengthened to yield $\II\lt[ E : T_{v} \rt] = \Omega\lt( n^2 \rt)$ (as argued in Proposition~4.4 of \cite{izumi2017triangle}), and this holds for any node $v$. Analogously to \eqref{eq:ig3}, we obtain 
$\II\lt[ E : \Pi_v \ \md|\ S_v \rt] = \Omega\lt( n^2 \rt).$
Before we can instantiate Lemma~\ref{lem:info_wc}(i) with $V_A=U=\set{v}$, $Z=E$, $V_B=\partial U = V(G) \setminus \set{v}$, $d_U=n-1$, and $L= \Theta\lt( n^{2} \rt)$, we need to verify that $|\supp(Z)| \le 2^{c \cdot L} = 2^{c \cdot \Theta\lt( n^2 \rt)}$, which indeed holds for a suitable constant $c$. 
This shows that $v$ requires $\Omega\lt( n / \log n \rt)$ awake rounds in expectation.
Since the above holds for any $v$, the claimed bound on the node-averaged energy complexity follows.
\end{proof}

\begin{corollary} \label{cor:scliques}
For $s$-clique listing, where every $s$-clique in the graph is output
by at least one of its nodes, the worst case energy complexity is $\Omega\lt( n^{1-2/s}/\log n \rt)$
rounds for algorithms that fail with some small constant probability $\epsilon > 0$. 
\end{corollary}
\begin{proof} 
As explained in \cite{fischer2018possibilities}, the proof is analogous to the lower bound for triangle enumeration by \cite{izumi2017triangle}, with the crucial difference being the use of  Lemma~1.3 of \cite{fischer2018possibilities} to bound the number of $s$-cliques in any graph with a given number of edges.
\end{proof}

We now turn our attention to graph problems for which the existing lower bound constructions are based on reductions from the set disjointness function.
By virtue of Corollary~\ref{cor:disj}, we immediately obtain lower bounds on the energy complexity of several important graph problems:

\begin{corollary} \label{cor:itcs} 
Computing an exact minimum vertex cover, maximum independent set, or minimum dominating set has 
a node-averaged energy complexity of $\Omega\lt( \frac{n^2}{\log^2 n} \rt)$ in expectation.
\end{corollary}
\begin{proof} 
In \cite{itcs24}, it is shown that these problems admit an $\ell$-separated lower bound family with respect to the set disjointness function $\disj_{\Theta\lt( n^2 \rt)}$, for $\ell = \Theta\lt( n / \log n \rt)$.
The result follows from Corollary~\ref{cor:disj}.
\end{proof}

\begin{corollary} \label{cor:diameter}
Computing the exact network diameter with $\epsilon$ error, for some small constant $\epsilon>0$, has a node-averaged energy complexity of $\Omega\lt( n/\log n \rt)$ in expectation. The same bound holds for computing exact all-pair-shortest-paths (APSP) routing tables.
\end{corollary}
\begin{proof} 
It is straightforward to verify that the lower bound construction of \cite{frischknecht2012networks} has all the important properties of a $2$-separated graph family with respect to $\disj_{\Theta\lt( n^2 \rt)}$ and the graph predicate $P$ for deciding whether the diameter is at most $4$.
Since there are $\Theta\lt( n \rt)$-nodes in $|V_1^+|$ and $|V_2^-|$  that form a perfect matching, we know that $d_{cut}=1$, and the maximum cut size is $s = \Theta\lt( n \rt)$.
Applying Corollary~\ref{cor:disj} completes the proof. 
\end{proof}

\begin{corollary} \label{cor:mwc}
In any undirected, weighted graph $G$, or directed, unweighted graph $G$, computing the exact Minimum Weight Cycle with $\epsilon$ error, for some small constant $\epsilon>0$, has a node-averaged energy complexity of $\Omega\lt( n/\log n \rt)$ in expectation. 
\end{corollary}
\begin{proof} 
It is straightforward to verify that the lower bound constructions of \cite{MR22} both have all the important properties of a $2$-separated graph family with respect to $\disj_{\Theta\lt( n^2 \rt)}$ and the graph predicate $P$ for deciding, respectively, whether there exists a directed cycle of length at most $4$ and a cycle of weight at most $6$, whereby the cut set size (for both constructions) is $s=\Theta\lt( n \rt)$.
Moreover, the edges in the (only) cut form a perfect matching and hence $d_{cut}=1$, and we can apply Corollary~\ref{cor:disj} to complete the proof. 
\end{proof}

\section{Conclusion}

In this paper, we presented an information-theoretic framework
for showing polynomial (in $n$) lower bounds for various fundamental graph problems. These energy lower bounds almost match (up to logarithmic factors) their respective round lower bounds. Hence, it is not fruitful to design distributed algorithms for these problems that achieve significantly lower energy complexity than their round complexity.

The good news is that our technique can be used to establish energy lower bounds, provided we have suitable communication-complexity lower bounds that are sufficiently large (e.g., polynomial in $n$).
There are still several problems, 
such as minimum cut and maximum matching, where we do not have (almost) matching energy bounds, which are worth studying.

\bibliographystyle{plain}
\bibliography{bibliography}

\newpage
\appendix
\section*{Appendix}
\section{Proof of Lemma~\ref{lem:info_wc} (Cut-based Energy Lower Bound)} \label{sec:info_wc_proof}

Consider an algorithm $\mathcal{A}$ that satisfies the premise of the lemma. 
In our analysis, we assume that $|\supp(Z)|\le 2^{c\cdot L}$ and mainly focus on proving (i). 
Wherever necessary, we explain how to make the argument work for (ii) (i.e., without the upper bound on the support size). 

\begin{definition} \label{def:extended_transcript}
The \emph{extended transcript} ${\hat{\Pi}}$ consists of the (non-empty) messages that the nodes in $U$ receive from $\partial U$, whereby we augment each message with the corresponding round number in which it was sent.
Formally, let random variable $M$ be the total number of messages the nodes in $U$ receive over all rounds from their neighbors in $\partial U$.
We define 
\begin{align}
\hat{\Pi}&=((R_1,\hat{\Pi}_1),\dots,(R_{M},\hat{\Pi}_{M})).\notag
\end{align}
\begin{itemize} 
\item Random variables $R_1 \le \cdots \le R_{M}$ are specific round numbers; for convenience, we define $R_0=0$.
\item Each $\hat{\Pi}_i$ corresponds to a non-empty message $m$ that a node $u \in U$ receives from a node in $\partial U$ during round $R_i$. 
Note that $\hat{\Pi}_i$ includes the source ID and destination ID.  
\end{itemize}
\end{definition}

\noindent 
See Figure~\ref{fig:extended_transcript} for an example of an extended transcript.

Note that it is perfectly possible (and even likely for some algorithms) that a node in $U$ wakes up in some round without receiving any message from a neighbor in $\partial U$. 
We define the \emph{sequence of wasted awake rounds} $\mathcal{W}=\lt( (W_1,u_1),(W_2,u_2),\dots \rt)$, where $W_{i} \le W_{i+1}$, to be the sequence of round-node pairs, such that $u_i \in U$ wakes up in round $W_i$ without receiving any (non-empty) message from $\partial U$.
We cannot simply ignore such rounds, as a node in $U$ may infer some information about $Z$ also from the absence of any message from $\partial U$.
Nevertheless, the reason why we do not need to explicitly include the wasted awake rounds in the extended transcript $\hat{\Pi}$ is that the entire sequence can be computed from $\hat{\Pi}$ and the initial state of the nodes in $V_A$ (which includes the nodes in $U$).
The next lemma formalizes this intuition. Its proof follows by induction over the extended transcript:
 
\begin{lemma} \label{lem:wasted}
The sequence of wasted awake rounds $\mathcal{W}$ of the nodes in $U$ is a deterministic function of the extended transcript $\hat{\Pi}$ and the initial state $S_{V_A}$ of the nodes in $V_A$,  
\end{lemma}

\begin{proof}
The proof follows by induction over the extended transcript. 
To keep the notation simple, we assume that $R_1<\dots<R_M$; it is straightforward to extend the proof to the general case where $R_1\le \dots \le R_M$.

Formally, we show the following claim: For all $i \in [M]$, the prefix $(R_{\le i-1}, \hat{\Pi}_{\le i-1})$ of $\hat{\Pi}$ together with the initial state $S_{V_A}$ determine the prefix $(W_{\le j_i},u_{\le j_i})$ of $\mathcal{W}$, where $j_i$ is the largest index such that rounds $W_1\le \dots \le W_{j_i}< R_i$.

For the base case ($i=1$), $R_0=0$ and $\hat{\Pi}_0$ is empty by definition.
Thus, it will be sufficient to show that all wasted awake rounds occurring prior to round $R_1$, i.e., before the first message from $\partial U$ is received by some node in $U$, are fully determined by the initial state $S_{V_A}$. For every $u \in U$, the timing of $u$'s next wake up step is a deterministic function of $S_{V_A}$, as this determines all communication that $u$ may receive from its neighbors in $G \setminus \partial U$, all of which are in $V_A$.
Thus the prefix $(W_{\le j_1},u_{\le j_1})$ of $\mathcal{W}$ is a function of $S_{V_A}$ as well.

For the induction step, assume that the statement holds for some $i \in [M]$, i.e., the prefix $(R_{\le i-1}, \hat{\Pi}_{\le i-1})$ and $S_{V_A}$ already determine all wasted awake steps $(W_{\le j_i},u_{\le j_i})$ prior to round $R_i$.
Our goal is to show that the additional knowledge of the pair $(R_{i}, \hat{\Pi}_{i})$ suffices to compute all wasted awake steps prior to $R_{i+1}$.
By the inductive hypothesis, the prefix $(W_{\le j_i},u_{j_i})$ up to the end of round $R_i-1$ is fully determined by  $S_{V_A}$ and the prefix $(R_{\le i-1}, \hat{\Pi}_{\le i-1})$.
Clearly, this is also sufficient to determine the state of each node in $V_A$ at the start of round $R_i$. 
Consequently, knowing $(R_{i}, \hat{\Pi}_{i})$ allows us to deterministically compute the state (and wake up steps) of all nodes in $V_A$ during round $R_i$, and, in fact, the same applies to all subsequent rounds in which no message is received by a node in $U$ from $\partial U$.
Hence, we have sufficient information to compute the states of the nodes in $V_A$ up until the end of round $R_{i+1}-1$.
\end{proof}

Since $\Pi_U$ denotes only the transcript of the messages received by nodes in $U$ from $\partial U$, whereas the extended transcript $\hat{\Pi}$ contains these same messages together with their respective round numbers and source/destination IDs, $\Pi_U$ is deterministically recoverable from $\hat{\Pi}$.
Thus, by the data processing inequality (see Fact~\ref{f:data_processing}), we have that
\begin{align}
 \II\lt[ \hat{\Pi} : Z \ \md|\ S_{V_A} \rt]
 \ge 
 \II\lt[ \Pi_U : Z \ \md|\ S_{V_A} \rt]
 \ge
 L. \label{eq:ent_lb}
\end{align}
Before deriving an upper bound on the mutual information in \eqref{eq:ent_lb}, we first define some additional notation.
In Figure~\ref{fig:vars}, we provide a list of variables used throughout the proof.
We use $X_{<i}$ as a shorthand for $X_{1},\ldots,X_{{i-1}}$, when considering a sequence of random variables, and we define $X_{\le i}$ similarly.
Below, we use $|X|$ to denote the worst case length of an optimal encoding of a random variable $X$, i.e., $|X| \le \log_2 |\supp(X)|$.

We first state some bounds that we rely on in our analysis. 
By the premise of the lemma, we have 
\begin{align}
 L \le n^{c_1}, \label{eq:L}
\end{align}
for some constant $c_1>0$.
Moreover, recall that we assume that the message size is at most logarithmic, i.e., there is some constant $c_2>0$ such that, for every $i$, we have 
\begin{align}
|\hat{\Pi}_i| \le c_2\log n. \label{eq:pihati}
\end{align}
Recall that $c>0$ is the constant specified in Property~(i) of the lemma.
We also introduce a constant $C$, where 
\begin{align}
	  C &\ge 4c(c_1 + c_2 + 2).  \label{eq:c}
\end{align}
Finally, we define the bounds $\beta$ and $\bar{M}$ such that 
\begin{align}
    \beta &= \frac{L}{ C \cdot d_U \log n }, \label{eq:beta}\\
    \bar{M} &= 2c\cdot \beta \cdot d_U. \label{eq:m_up}
\end{align}
Note that we assume $L = \omega(d_U \log n)$ throughout the proof, as otherwise the resulting lower bound is trivial.

\begin{figure}
\small
\centering
    \begin{tabular}{l l l}
        Variable & Description  \\
        \hline
        $\alpha_U$ & total number of awake rounds summed over the nodes in $U$\\
        $\beta$ & assumed upper bound on $\alpha_U$\\
        $d_U$ & maximum number of neighbors in $\partial U$ of any node in $U$\\
        $e(U,\partial U)$ & the number of edges across the cut\\
        $\mathcal{E}$, $\mathbf{1}_{\mathcal{E}}$ & event $\mathcal{\mathcal{E}}$ occurs if $M \!\le\! \bar{M}$; $\mathbf{1}_{\mathcal{E}}$ is its indicator random variable
\\
        $L$ & amount of information that the nodes in $U$ need to learn about $\hat{\Pi}$ \\
        $\bar{M}$ & upper bound on $M$\\
        ${M}$ & actual length of the extended transcript $\hat{\Pi}$ \\
        $\hat{\Pi}$ & extended transcript (see Def.~\ref{def:extended_transcript})\\
        ${\Pi_U}$ & transcript of messages sent from $\partial U$ to $U$\\
        $(R_i,\hat{\Pi}_i)$ & message $\hat{\Pi}_i$ was sent in round $R_i$ \\
        $S_{V_A}$ & initial knowledge of the nodes in $V_A$\\
        $U$ & nodes in $V_A$ that have an edge to $V_B$\\
        $\partial U$ & nodes in $V_B$ that have an edge to $V_A$\\
        $(V_A,V_B)$ & vertex cut that we are focusing on\\
        $Z$ & function of the graph about which the nodes in $U$ need to learn something\\
        \hline\\ 
    \end{tabular}
\small
    \caption{List of variables and parameters used in the proof of Lemma~\ref{lem:info_wc}.}
    \label{fig:vars}
\end{figure}

We are now ready to start our analysis. Assume towards a contradiction that the given algorithm satisfies 
\begin{align}
\EE\lt[ \alpha_U \rt] \le \beta. \label{eq:beta_ub}
\end{align}

Whenever a node $u \in U$ is awake in a given round, it may receive messages from up to $e(u,\partial U) \le d_U$ neighbors in $\partial U$.
Thus, we can bound the expectation of the number $M$ of pairs $(R_i,\hat{\Pi}_i)$ in $\hat{\Pi}$ as $\EE\lt[ M \rt] \le \EE\lt[ \alpha_U \rt] \cdot d_U$.
Notice that the event $M \le \bar{M}$ occurs with probability at least $1 - \frac{1}{2c}$ by Markov's inequality. 
Let $\mathcal{\mathcal{E}}$ denote the event $M \!\le\! \bar{M}$ and let $\mathbf{1}_{\mathcal{E}}$ be the indicator random variable for this event.
By Fact~\ref{f:leq},
\begin{align}
\II\lt[ \hat{\Pi} : Z \ \md|\ S_{V_A} \rt] 
&\le
\II\lt[ \hat{\Pi} : \mathbf{1}_{\mathcal{E}},Z \ \md|\ S_{V_A} \rt] \notag\\ 
\ann{by Fact~\ref{f:Ichain}}
&=
\II\lt[ \hat{\Pi} : Z \ \md|\ S_{V_A},\mathbf{1}_{\mathcal{E}} \rt]
+
\II\lt[ \hat{\Pi} : \mathbf{1}_{\mathcal{E}} \ \md|\ S_{V_A} \rt] \notag\\ 
\ann{by Fact~\ref{f:len}}
&\le
\II\lt[ \hat{\Pi} : Z \ \md|\ S_{V_A},\mathbf{1}_{\mathcal{E}} \rt] 
+
1 \notag\\ 
\ann{since $\neg\mathcal{E}$ is $\set{M > \bar{M}}$}
&= 
\Pr\lt[ \mathbf{1}_{\mathcal{E}}\!=\! 1 \rt] 
\II\lt[ \hat{\Pi} : Z \ \md|\ S_{V_A},\mathcal{E} \rt] 
+
\Pr\lt[ \mathbf{1}_{\mathcal{E}}\!=\! 0 \rt] 
\II\lt[ \hat{\Pi} : Z \ \md|\ S_{V_A},M > \bar{M} \rt] 
+
1 \notag\\ 
&\le
\II\lt[ \hat{\Pi} : Z \ \md|\ S_{V_A},\mathcal{E} \rt] 
+
\frac{1}{2c}
\HH\lt[ Z \ \md|\ M > \bar{M} \rt] 
+
1 \notag\\ 
\ann{by Fact~\ref{f:supp}}
&\le
\II\lt[ \hat{\Pi} : Z \ \md|\ S_{V_A},\mathcal{E} \rt] 
+
\frac{1}{2c}
\log_2\lt( |\supp(Z  \mid M > \bar{M})| \rt)
+
1,\notag \\
&\le 
\II\lt[ \hat{\Pi} : Z \ \md|\ S_{V_A},\mathcal{E} \rt] 
+
\frac{1}{2c}
\log_2\lt( |\supp(Z)| \rt)
+
1,\notag 
\end{align}
where, in the final inequality, we have used the fact that conditioning on an event cannot increase the support of a random variable.\footnote{To see that $|\supp(Z \mid \mathcal{E})| \le |\supp(Z)|$, for any event $\mathcal{E}$, it suffices to observe that $\forall z\colon \Pr\lt[ Z\!=\! z \ \md|\ \mathcal{E} \rt]  = \frac{\Pr\lt[ Z= z \wedge \mathcal{E} \rt] }{\Pr\lt[ \mathcal{E} \rt] } \le \frac{\Pr\lt[ Z= z\rt] }{\Pr\lt[ \mathcal{E} \rt] }.$ Note that conditioning on $\mathcal{E}$ may increase the entropy of $Z$ nevertheless.}
Applying the upper bound on $|\supp(Z)|$ guaranteed by the premise in (i), we obtain 
\begin{align}
\II\lt[ \hat{\Pi} : Z \ \md|\ S_{V_A} \rt] 
\le
\II\lt[ \hat{\Pi} : Z \ \md|\ S_{V_A},\mathcal{E} \rt] + \frac{L}{2} + 1. \label{eq:mutual_ub1}
\end{align}
To see why \eqref{eq:mutual_ub1} also holds for algorithms satisfying only (ii) and not the bound on the support in (i), observe that the event $M \le \beta \cdot d_U \le \bar{M}$ holds with probability $1$ for algorithms that guarantee a worst case upper bound on the number of awake rounds.

For any $j \ge M+1$, we define $R_j=\infty$ and $\hat{\Pi}_j=\langle\rangle$ (i.e., the empty message), which in particular implies that $\HH\lt[ R_{j},\hat{\Pi}_{j} \rt]=0$.
Next, we derive an upper bound on the term $\II\lt[ \hat{\Pi} : Z \ \md|\ S_{V_A},\mathcal{E} \rt]$ on the right-hand side of \eqref{eq:mutual_ub1}. 
We get 
\begin{align}
\II\lt[ \hat{\Pi} : Z \ \md|\ S_{V_A},\mathcal{\mathcal{E}} \rt]
 &\le
 \HH\lt[ \hat{\Pi} \ \md|\ S_{V_A},\mathcal{\mathcal{E}} \rt] \notag\\
 \ann{by Fact~\ref{f:Hchain} and Def.~\ref{def:extended_transcript}}
 &=
 \sum_{i=1}^{\bar{M}}  
 \HH\lt[ R_{i},\hat{\Pi}_{i} \ \md|\ R_{<i},\hat{\Pi}_{<i},S_{V_A},\mathcal{\mathcal{E}} \rt] 
 \label{eq:ent_ub}
\end{align}
Returning to \eqref{eq:mutual_ub1}, we get
\begin{align}
\II\lt[ \hat{\Pi} : Z \ \md|\ S_{V_A} \rt]
 &\le
 \sum_{i=1}^{\bar{M}}  \HH\lt[ R_{i},\hat{\Pi}_{i} \ \md|\ R_{<i},\hat{\Pi}_{<i},S_{V_A},\mathcal{E} \rt] + \frac{L}{2} 
  + 1. \label{eq:ent_ub2}
\end{align}

Combining the upper bound in \eqref{eq:ent_ub2} with the lower bound in \eqref{eq:ent_lb} shows that there must exist an index $i^*$ such that
\begin{align}
\HH\lt[ R_{{i^*}},\hat{\Pi}_{{i^*}} \ \md|\ R_{<i^*},\hat{\Pi}_{<i^*},S_{V_A},\mathcal{E}  \rt] 
 &\ge
 \frac{L/2  - 1 }{\bar{M}} \notag\\ 
 \ann{by \eqref{eq:m_up}}
 &=
 \frac{L/2}{2c\beta \cdot d_U}  - o(1)\notag\\ 
 \ann{by \eqref{eq:beta}}
 &=
 \frac{C}{4c}\log n - o(1).
 \label{eq:ent_lb2}
\end{align}
Without loss of generality, assume that $i^*$ is the smallest such index.
Rewriting the entropy on the left-hand side via the chain rule (see Fact~\ref{f:Hchain}) yields 
\begin{align}
\HH\lt[ R_{i^*},\hat{\Pi}_{i^*} \ \md|\ R_{<i^*},\hat{\Pi}_{<i^*},S_{V_A},\mathcal{E}  \rt] 
  &= 
    \HH\lt[ R_{i^*}\ \md|\ R_{<i^*},\hat{\Pi}_{<i^*},S_{V_A},\mathcal{E}  \rt] \notag\\ 
    &\phantom{---}
    +
    \HH\lt[ \hat{\Pi}_{i^*}\ \md|\  R_{i^*},R_{<i^*},\hat{\Pi}_{<i^*},S_{V_A},\mathcal{E}   \rt] \notag\\ 
  \ann{by Fact~\ref{f:len}}
  &\le
  \HH\lt[ R_{i^*}\ \md|\ R_{<i^*},\hat{\Pi}_{<i^*},S_{V_A},\mathcal{E}   \rt] 
  +
   |\hat{\Pi}_{i^*}| \notag\\ 
   \ann{by \eqref{eq:pihati}}
  &\le
\HH\lt[ R_{i^*}\ \md|\ R_{<i^*},\hat{\Pi}_{<i^*},S_{V_A},\mathcal{E}   \rt] 
  +
  c_2 \log n.\notag
\end{align}
Due to the lower bound in \eqref{eq:ent_lb2}, we get
\begin{align}
\HH\lt[ R_{i^*}\ \md|\ R_{<i^*},\hat{\Pi}_{<i^*},S_{V_A},\mathcal{E}   \rt] 
\ge 
 \lt( \frac{C}{4c} - c_2\rt)\log n - o(1).
  \label{eq:r_ent}
\end{align} 

Let $U_* \subseteq U$ be the nodes that receive a message from some $v \in \partial U$ in round $R_{i^*}$. (Note that it may be the case that $R_{i^*} = R_{i^*+1} = \dots = R_{i^*+\ell}$, and in that case, multiple nodes may receive a message from their neighbors in $\partial U$ in the same round.)
The intuition behind \eqref{eq:r_ent} is that the exact value of $R_{i^*}$ is not known to the nodes in $U_*$ and thus they may need to guess a number of times before one of them wakes up in the correct round.
Of course, when ``guessing'' $R_{i^*}$, the nodes in $U_*$ may take into account all earlier received messages and the respective rounds in which they were sent, as well as the earlier unsuccessful wake up attempts in which they did not receive any message.
This is represented by the conditioning on events $R_{<i^*}\!=\! r,\hat{\Pi}_{<i^*}\!=\! \pi,S_{V_A}\!=\! s$, and $\mathcal{E}$. 
Note that we provide more power to the algorithm by assuming that all nodes in $U_*$ know $S_{V_A}$, as well as all messages received so far by any node in $U$.
In fact, random variable $R' := (R_{i^*} \mid R_{<i^*}\!=\! r,\hat{\Pi}_{<i^*}\!=\! \pi,S_{V_A}\!=\! s,\mathcal{E})$ includes conditioning on the (joint) knowledge of the nodes in $U_*$ (as well as the initial states of all other nodes in $V_A$), and assumes that---in contrast to reality---the nodes even know that event $\mathcal{E}$ occurs, which can only strengthen the lower bound.
In addition, every node in $U_*$ knows, in round $r$, all earlier wake up steps of other nodes in $U$, in which they did not receive any message from $\partial U$, which is justified due to Lemma~\ref{lem:wasted}.
Due to the assumption that $(U,\partial U)$ is an edge cut, it is clear that a node $u \in U_*$ cannot gain any information about $R_{<i^*}$ by receiving a message from other nodes in $U$ (or $V_A$), as that knowledge is already part of the given conditioning. 
Thus, we can interpret each subsequent wake-up of $u$ in some round $k$ as a single guess of the form ``$R_{i^*} = k$?''. 
The work of Massey~\cite{massey1994guessing} shows that guessing the value of a random variable $X$ in this manner requires at least $2^{\HH\lt[ X \rt]  - 2}+1$ guesses in expectation.
Let $T$ be the total number of times that some node in $U_*$ wakes up until one of them succeeds in guessing the correct round $R_{i^*}$.
Applying the aforementioned result of \cite{massey1994guessing}, we get
\begin{align}
 \EE\lt[ T\ \md|\ R_{<i^*}\!=\! r,\hat{\Pi}_{<i^*}\!=\! \pi,S_{V_A}\!=\! s,\mathcal{E}  \rt] 
 &= \Omega\lt( 2^{\HH\lt[ R_{i^*}\ \md|\ R_{<i^*}= r,\hat{\Pi}_{<i^*}= \pi,S_{V_A}= s,\mathcal{E}   \rt]} \rt). \label{eq:guesses}
\end{align}
Let $\mathbf{1}_{\mathcal{E}}$ be the indicator random variable for event $\mathcal{E}$, which is $1$ if $M \le \bar{M}$.
To shorten the notation, we define random variable $\mathsf{Y}=(R_{<i^*},\hat{\Pi}_{<i^*},S_{V_A})$.
We have
\begin{align}
\EE\lt[ \alpha_U \rt] \ge
\EE\lt[ T \rt] 
&=
  \EE\limits_{y\sim\mathsf{Y},b\sim\mathbf{1}_{\mathcal{E}}}\Big[ \EE\lt[\ T\ \md|\ \mathsf{Y} \!=\! y,\mathbf{1}_{\mathcal{E}}\!=\! b  \rt] \Big] \notag\\ 
\ann{since $\Pr\lt[ \mathcal{E} \rt] \ge 1-\frac{1}{2c}$}
&\ge
\lt( 1 - \frac{1}{2c} \rt)\EE\Big[ \EE\lt[\ T\ \md|\ \mathsf{Y} \!=\! y,\mathcal{E}  \rt] \ \Big|\ \mathcal{E} \Big] \notag\\ 
  \ann{by \eqref{eq:guesses}} 
&\ge
\lt( 1 - \frac{1}{2c} \rt)\EE\Big[ 2^{\lt(\HH\lt[ R_{i^*}\ \md|\ \mathsf{Y}= y,\ \mathcal{E}\   \rt]\ -\ 2\rt)} \ \Big|\ \mathcal{E} \Big] \notag\\ 
  \ann{by Jensen's inequality}
&\ge
\lt( 1 - \frac{1}{2c} \rt) 2^{\lt(\EE\lt[\HH\lt[ R_{i^*}\ \md|\ \mathsf{Y}= y,\ \mathcal{E}\   \rt] \ \md|\ \mathcal{E}\ \rt]\ -\ 2\rt)}  \notag\\ 
\ann{by \eqref{eq:cond_ent}}
&\ge
\lt( 1 - \frac{1}{2c} \rt) 2^{\lt(\HH\lt[ R_{i^*}\ \md|\ \mathsf{Y}, \mathcal{E}\   \rt]\  -\ 2\rt)}  \notag\\ 
\ann{by \eqref{eq:r_ent}}
&\ge
\lt( 1 - \frac{1}{2c} \rt) 2^{\lt( \frac{C}{4c} - c_2\rt)\log n - 2 - o(1)}  \notag\\ 
&=
\Omega\lt(  2^{\lt( \frac{C}{4c} - c_2\rt)\log n - 3}  \rt) \notag\\ 
\ann{by \eqref{eq:c}}
&=
\Omega\lt( n^{c_1+2} \rt).
\end{align}
However, according to \eqref{eq:L}, we have $L \le n^{c_1}$, and thus we have arrived at a contradiction to the assumed upper bound on $\EE\lt[ \alpha_U \rt]$, which completes the proof of Property~(i).
Similarly to the above, we can conclude that Property~(ii) also holds, due to the event $\mathcal{E}$ occurring with probability $1$ for algorithms that have a (per node) worst case upper bound on the number of awake rounds.
This completes the proof of Lemma~\ref{lem:info_wc}.

\end{document}